\documentclass[12pt]{article}

\usepackage{url}
\usepackage{amssymb}
\usepackage{pstricks}
\usepackage{amsmath}
\usepackage{amsfonts}
\usepackage[onehalfspacing,nodisplayskipstretch]{setspace}
\usepackage[a4paper,nohead,left=1in,right=1in,top=1in,bottom=1in]{geometry}
\usepackage{graphicx}
\usepackage{array}
\usepackage{cancel}
\usepackage{tikz}
\usepackage{bbm}
\title{Does Your Blockchain Need Multidimensional Transaction Fees?}
\author{Nir Lavee \footnote{Tel Aviv University and Starkware.} \and Noam Nisan\footnote{Hebrew University and Starkware. This work was supported by the Israeli Science Foundation (ISF number 505/23).} \and Mallesh Pai\footnote{Special Mechanisms Group, Consensys and Paradigm.} \and Max Resnick \footnote{Anza Labs and Special Mechanisms Group, Consensys.}}

\newtheorem{theorem}{Theorem}
\newtheorem{observation}{Observation}

\newtheorem{claim}[theorem]{Claim}

\newtheorem{corollary}[theorem]{Corollary}

\newtheorem{definition}{Definition}
\newtheorem{example}{Example}

\newtheorem{remark}[theorem]{Remark}

\newenvironment{proof}[1][Proof]{\noindent\textbf{#1.} }{\ \rule{0.5em}{0.5em}}
\let\oldexample\example
\renewcommand{\example}{\oldexample\normalfont}
\let\oldremark\remark
\renewcommand{\remark}{\oldremark\normalfont}

\newenvironment{ignore}[1]{}

\newcommand{\norm}[1]{\left\lVert#1\right\rVert}

\newcommand{\Rnonneg}{\mathbb{R}_{\ge 0}} % Non-negative reals

\begin{document}

\maketitle

\begin{abstract}
Blockchains have block‑size limits to ensure the entire cluster can keep up with the tip of the chain. These block-size limits are usually single-dimensional, but richer multi‑dimensional constraints allow for greater throughput. The potential for performance improvements from multi-dimensional resource pricing has been discussed in the literature, but exactly how big those performance improvements are remains unclear. In order to identify the magnitude of additional throughput that can be unlocked by multi-dimensional transaction fees, we introduce the concept of an $\alpha$‑approximation. A constraint set $C_1$ is $\alpha$‑approximated by $C_2$ if every block feasible under $C_1$ is also feasible under $C_2$ once all resource capacities are scaled by a factor of $\alpha$ (e.g., $\alpha =2$ corresponds to doubling all available resources). We show that the $\alpha$-approximation of the  optimal single-dimensional gas measure corresponds to the value of a specific zero sum game.  However, the more general problem of finding the optimal $k$-dimensional approximation is NP-complete. Quantifying the additional throughput that multi-dimensional fees can provide allows blockchain designers to make informed decisions about whether the additional capacity unlocked by multi-dimensional constraints is worth the additional complexity they add to the protocol. 

\end{abstract}

\section{Introduction and Motivation}

Any new block added to the blockchain has to be replayed by the entire network. If the block is too large and takes too long to replay, then slower replicas may fall behind the tip of the chain, leading to degraded performance or even a pause in block production until the slow replicas can catch up. To ensure that all replicas can keep up, blockchain designers place constraints on which transactions can be packed into a valid block. If the leader proposes a block that violates these constraints, the rest of the cluster realizes that the block is invalid and rejects it. 

The original Bitcoin implementation \cite{bitcoin} defined size as the number of bytes and limited a block to 1 Megabyte.\footnote{The later ``SegWit'' upgrade of Bitcoin took into account that not all bytes impose the same burden on the implementation, defined different ``weights'' for different components of a transaction, giving a measure of ``virtual byte'', and limited blocks to 4 Mega {\em Virtual} Byte. This typically allowed 60\%-70\% more transactions in a block relative to raw bytesize measurements.} Most blockchains in production today are only slightly more sophisticated than this. Nearly all use a single-dimensional measure of block size: for example, until recently, Ethereum sized all transactions according to a measure it calls gas, and Solana measures transactions according to Computational Units (CUs). But blockchains have more than a single constraint on new block production. Not only do blockchains inherit all of the intricacies of the computers that run their client software (e.g. storage, computation, memory bandwidth) but they also have to grapple with communication between the replicas (networking). Therefore, a blockchain protocol designer may wish to meter and price each of these resources separately---for example, certain transactions may be computationally trivial but require large amounts of networking or storage. This is not merely a theoretical concern---as of the Dencun upgrade in Spring 2024, Ethereum now defines two types of resources ``gas'' and ``blobs''.  At the time of writing of this paper, Ethereum allows blocks to use, at a maximum, 36M gas and 6 blobs.\footnote{Ethereum also has a targeted size of 18M gas and 3 blobs, and the size of the block relative to the target affects the prices charged in the next block. We will ignore the target and focus solely on the maximum as the target is orthogonal to this paper.} 

While there has been some recent work dealing with multi-dimensional constraints in blockchain systems \cite{MDBFM23, CMW23, ADM24}, handling multi-dimensional constraints remains difficult.  Not only does the optimal assembly of blocks become complex algorithmically, but the pricing of transactions and the bidding process --- the transaction fee mechanism --- must also take into account multiple resources, which further complicates the strategic analysis. Avoiding this complexity is very desirable. This paper asks the question of whether and to what extent the complexity of multiple dimensions is worth the cost: how much do we lose by sticking to a single resource dimension?

Take, for example, a blockchain that has two resources, gas and blobs, similar to Ethereum. Any given transaction can use some amount of gas and blobs. Can we replace the two restrictions of $\texttt{gas} \le 36\text{M}$
and $\texttt{\texttt{blobs}}\le 6$ by some single restriction $(\alpha \cdot \texttt{gas} + \beta \cdot \texttt{blobs} \le B)$?
It is not difficult to see that the two constraints cannot be replaced by a single linear one  without incurring a loss.  Figure \ref{f1} depicts the two ``true'' constraints as well as a single linear inequality that implies both of them ($\texttt{gas} + 6 \cdot \texttt{blobs} \leq 36$. However, note that this  unavoidably rules out much of the space of actual block capacities in the worst case: a block which uses $18M$ \texttt{gas} and $3$ blobs is considered as full by the linear measure, whereas it is clearly half-full on each dimension. 

\begin{figure}\label{f1}
    \centering
    \begin{tikzpicture}[scale=0.7]
    % Draw axes
    \draw[->, thick] (0,0) -- (9,0) node[right] {Gas used};
    \draw[->, thick] (0,0) -- (0,4) node[above] {Blobs used};
    
    % Constraints
    \draw[thick] (0,3) -- (8,3) node[midway, above] {$y \leq 6$};
    \draw[thick] (8,0) -- (8,3) node[midway, right] {$x \leq 36$};
    \draw[thick] (0,3) -- (8,0) node[midway, above right] {$x + 6y \leq 30$};
    
    % Labels
    \node[left] at (0,3) {6};
    \node[below] at (8,0) {36};
    
\end{tikzpicture}
    \caption{Two Ethereum-like 
constraints: the quantity $x$ of one resource (gas) limited by $x \le 30$
and the quantity $y$ of another resource (fractional blobs) 
limited by $y \le 6$.  The 
best single-dimensional gas measure that captures both constraints is $x+5y \le 30$, losing a factor of 2 in capacity in the worst case.}
    \label{fig:1}
\end{figure}
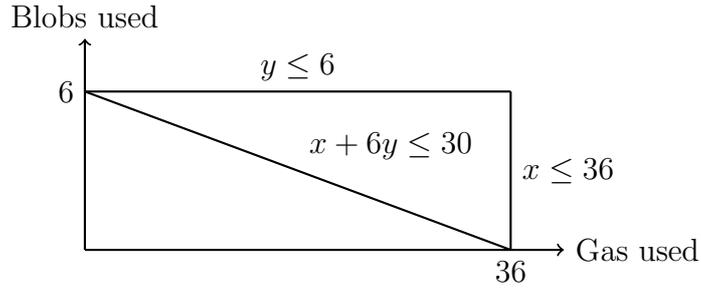

It is easy to verify that, in the worst case, a single linear constraint may lose a factor of $k$ in a setting with $k$ distinct resources. But this is without any further modeling of how various resources are used by basic operations available in the protocol. How much would actually lose given a specific instruction set?  

We suggest a way to quantify the loss incurred based on the specific instruction set and its associated costs, allowing a more nuanced decision of whether or not the complexity of multi-dimensional fees is required. The hope is that in cases where a single gas measure will incur only a small loss of capacity, blockchain designers will choose simplicity. Conversely, in other cases, the loss may be too large, and the system would be better off with a multi-dimensional fee market.

\section{Model}
%We consider the following simple model. There is a set of resources. For each resource $j$ there is a resource use limit of $B_j$ per block. %There is a finite set of transactions (think of each transaction as e.g. an atomic opcode). Each transaction $t$ uses some amount $a_{tj}$ of each resource $j$.  Thus, a set of transactions $T$ ``fits'' into a single block, if for every resource $j$ it holds that $\sum_{t \in T} a_{tj} \le B_j$. Our question is straightforward: given the transactions and their resource usage patterns, what is the loss (if any) from considering a single-dimensional gas measure? 

\begin{figure}[h]
    \centering
    \begin{tikzpicture}
        % Main matrix (light blue)
        \fill[blue!10] (0,-1) rectangle (4,3);
        
        % Left side rectangle (operations)
        \fill[red!10] (-1.2,-1) rectangle (-0.2,3);
        
        % Top rectangle (resources)
        \fill[green!10] (0,3.2) rectangle (4,4.2);
        
        %block
        \fill[purple!20] (-5.5,0.5) rectangle (-1.5,1.5);
        %constraint
        \fill[teal!20] (5,0.5) rectangle (9,1.5);
        
        % Labels
        \node at (-0.7,1) [rotate=90] {\textbf{Operations}};
        \node at (2,3.7) {\textbf{Resources}};
        \node at (-3.5,1)  {\textbf{Block}};
        \node at (7,1)  {\textbf{$B = (B_j)$}};
        % Matrix content representation
        \node at (2,1) {$W = (w_{ij})$};
        % Dot product symbol between blue and purple rectangles
        \node at (-1.4,1) {$\cdot$};
        % Capacity constraints
        \node at (4.5,1) {$\leq$};

    \end{tikzpicture}
    \caption{Each row of the Operations-Resource Matrix represents an operation, each column represents a resource, and each entry $w_{ij}$ represents the amount of resource $j$ consumed by operation $i$.}
    \label{fig:resource-matrix}
\end{figure}
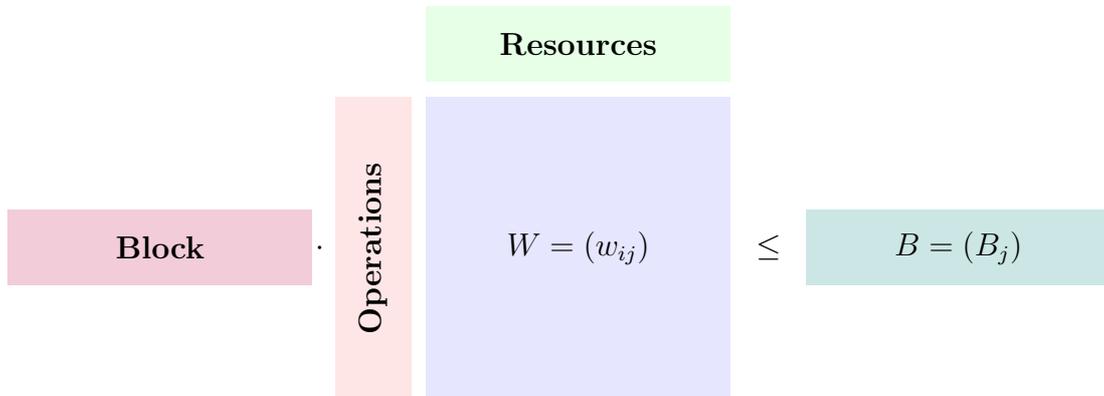

The basic property of the system that we will consider is the set of {\em primitive operations} that the system provides (e.g. Ethereum opcodes or Solana Instructions), where each transaction is composed of a sequence of primitive operations. Each of these primitive operations $i \in I$ is associated with the amount $w_{ij}$ that it uses of each resource $j \in J$. The actual constraints of the system allow a block to contain at most $B_j$ quantity of each resource $j$.  The parameters of the system are thus the operation-resource matrix $W=(w_{ij})$ and the resource capacities vector $B=(B_j)$. We model the total use of resources by a transaction as the sum of the resources used by all of the operations in it, so a set of transactions $T$ fits in a single block if and only if for every resource $j$ we have that $\sum_{i} x_i w_{ij} \le B_j$, where $x_i$ is the number of uses of primitive operation $i$ in $T$. We will now define the sense in which  a {\em single synthetic constraint} can replace the {\em set of actual constraints}. 

\begin{definition}
    A gas measure for a given set of primitive operations 
    is a non-negative vector $g=(g_i)$.
    The gas measure is said to represent a given operation-resource matrix 
    $W=(w_{ij})$ with resource capacities 
    $B=(B_j)$ if for any non-negative vector $x=(x_i)$ we have that
    $\sum_i x_i g_i \le 1$ implies that  
    $\sum_i x_i w_{ij} \le B_j$ for all $j$.
\end{definition}

This definition requires that a single  constraint on a set of transactions that requires the total gas measure to be at most $1$ is a {\em sufficient} condition for all resource constraints 
to be met.  It is easy to observe that
for every operation-resource matrix $W$ and resource capacities vector $B$ there is a
unique optimal representation as a single gas measure.  The simple proof appears in 
Appendix \ref{app-proof}.

\begin{observation}\label{obs:measure}
    For every operation-resource matrix 
    $W=(w_{ij})$ with resource capacities 
    $B=(B_j)$, the gas measure  given by $g_i = max_j w_{ij}/B_j$ represents
    $W$ with $B$.  Furthermore, $g$ is minimal in the sense that for
    any gas measure $g'$ that represents $W$ with $B$ we have that for
    all $i$, $g'_i \ge g_i$.
\end{observation}

We thus see that if we want a single synthetic gas measure
to imply all constraints in the worst case then 
the cost $g_i$
of each operation should be the maximum over all resources of the fraction of this resource's capacity that is taken by the operation.  

\section{Results}

The gas constraint being satisfied is, by definition, a sufficient condition for all resource constraints being satisfied. It is certainly not a necessary condition as Figure \ref{f1} illustrates.  But perhaps it is ``close'' to a necessary condition in the sense that some $\alpha$-factor relaxation of the 
gas constraint is a necessary condition?  Quantifying this closeness is what we are aiming for.
   
\begin{definition}
    A gas measure $g$ that represents an operation-resource matrix 
    $W=(w_{ij})$ with resource capacities $B=(B_j)$ is called an 
    $\alpha$-approximate representation ($\alpha \ge 1$), if
    for any non-negative vector $x=(x_i)$ we have that
    $\sum_j x_i w_{ij} \le B_j$ for all $j$ implies that 
    $\sum_i x_i g_i \le \alpha$.  
    
    The {\em single-dimensional approximability}
    of $W$ with $B$ is the smallest value $\alpha$ for
    which there exists an $\alpha$-approximate gas measure representation
    of $W$ with $B$.
\end{definition}

The level of approximability quantifies how much is lost when replacing the 
multidimensional constraints with the best single constraint that represents it.  We completely 
characterize this level as a value of the (zero-sum) game that is associated with $W$ and $B$.
Perhaps surprising at first sight, the proof is straightforward and appears in Appendix \ref{app-proof}.

\begin{theorem} \label{th:main}
    The single-dimensional
    approximability of an operation-resource matrix 
    $W=(w_{ij})$ with resource capacities $B=(B_j)$ is the exactly the reciprocal of
    the value
    of the zero-sum game with utilities $u_{ij} = w_{ij}/(B_j \cdot g_i)$
    and where $g_i = max_j w_{ij}/B_j$ is the gas measure achieving this approximation 
    (where the row player, who chooses $i$, is the minimizer.)
\end{theorem}

Table 1 presents an example that has two resources 
with maximum sizes 
$B_1=15$ and $B_2=3$ and four possible operations and 
calculates
the minimal
single gas measure $g$ that represents it.  
Table 2 presents the associated game with the player's 
min-max strategies and the associated value of the game which is the reciprocal of the approximation ratio, $\alpha=11/8$, that is achieved by $g$.

\begin{table}
\centering
\begin{minipage}{0.45\textwidth}
\centering
\begin{tabular}{|c||c|c||c|} 
\hline 
{\bf Operation} & {\bf $w_{i1}$} & {\bf $w_{i2}$} & $g_i$\\ 
\hline \hline
{\bf Op1}  & 2 & 1 & 1/3\\ 
\hline 
{\bf Op2}  & 6 & 2 & 2/3\\ 
\hline 
{\bf Op3}  & 9 & 1 & 3/5\\ 
\hline 
{\bf Op4}  & 10 & 1 & 2/3\\ 
\hline \hline
{\bf $B_j$}    & 15 & 3 &\\ 
\hline
\end{tabular}
\caption{\small An example with four operations and two resources: the operation matrix $W=(w_{ij})$, resource bounds $\vec{B}$, and the implied single gas measure $\vec{g}$.}
\label{tab:my_label1}
\end{minipage}\hfill
\begin{minipage}{0.45\textwidth}
\centering
\begin{tabular}{|c||c||c||c|} 
\hline 
{\bf Operation} & $u_{i1}$ & $u_{i2}$ & {\bf $x_i$}\\ 
\hline \hline
{\bf Op1}  & 2/5 & 1 & 5/11 \\ 
\hline 
{\bf Op2}  & 3/5 & 1 & 0 \\ 
\hline 
{\bf Op3}  & 1 & 5/9 & 0 \\ 
\hline 
{\bf Op4}  & 1 & 1/2 & 6/11 \\ 
\hline \hline
{\bf $y_j$}   & 5/11 & 6/11  & {\bf $1/\alpha = 8/11$} \\ 
\hline
\end{tabular}
\caption{\small The associated game $U=(u_{ij})$, column player (maximizer) strategy $\vec{y}$, 
and row player (minimizer) strategy $\vec{x}$ at equilibrium, and game value $1/\alpha$.}
\label{tab:my_label2}
\end{minipage}
\end{table}

The formal way to interpret an approximability ratio of $\alpha$ is as
a resource augmentation bound: suppose that we use the single gas measure $g$ instead of
the true multidimensional constraints but compensate by augmenting all resource capacities by a factor of $\alpha$. Then, for any set of transactions that can fit into an original block
according to the real multidimensional constraints, the single gas measure will allow 
them to be packed into a single $\alpha$-augmented block.  A more natural interpretation is that, if one assumes that  transaction sizes are small relative to block sizes (and so integrality constraints are insignificant),\footnote{Alternatively, as in \cite{BN25}, that the variable block size allowed by EIP-1559 can compensate for integrality constraints.} then the total capacity obtained by the single-dimensional measure is
at least $1/\alpha$ fraction of the maximum possible capacity allowed by the true
multidimensional constraints. 

This result compares the theoretical capacity of a blockchain under its true, multi‑dimensional constraints with that of a simplified model relying on a single gas‑measure constraint. Once the throughput gains provided by multi‑dimensional fees are clearly quantified, blockchain designers can more easily balance the trade‑off between the high throughput (but additional complexity) of multi‑dimensional fee markets and the simplicity (but lower capacity) of single‑dimensional fee markets.

\section{Optimal \(k\)-Dimensional Measures}
Our basic results correspond to the loss from a single-dimensional measure. What should a designer do if the computed worst-case loss from a single-dimensional measure is large? 

A reasonable next step might be to accept a $k$-dimensional measure where $k$ is much smaller than the true number of resources $|J|$ which we denote for brevity $n$.\footnote{Note that in reality, a system has many resources that can be used in parallel---network, cache storage, cold storage, multiple cores, GPUs, etc. Nevertheless, a planner may wish to only meter a few ``amalgamated'' resources such as computation and storage in order to mitigate the UX and design issues previewed above and discussed in the Appendix \ref{app-comp}.} In this section, we provide a preliminary investigation of selecting $k$ and the associated loss.

To that end, as before, let $
W = (w_{ij}) \in \mathbb{R}^{|I|\times n},$
where \(I\) indexes the set of primitive operations and \(j=1,\dots,n\) indexes the \(n\) resources. Let $B = (B_1, B_2, \ldots, B_n) \in \mathbb{R}_{>0}^n$ denote the resource capacities and define the normalized matrix
\[
w'_{ij} = \frac{w_{ij}}{B_j}.
\]
A block, represented by a nonnegative vector \(x = (x_i)\), of counts of primitive operations is \emph{feasible} if
\[
\sum_{i\in I} x_i\, w'_{ij} \le 1 \quad \text{for all } j = 1,\dots,n.
\]
Let $K$ be the largest possible 1-norm for any feasible block. 

When \(n\gg k\), we may wish to compress the \(n\) resources into a \(k\)-dimensional measure. 

\begin{definition}[$k$-Dimensional Gas Measure ]
A $k$-dimensional gas measure $A \in \Rnonneg^{|I|\times k}$  \textbf{represents} $W'$ if for any $x \ge 0$, we have:
\[ \forall r \leq q : \sum_i x_i A_{ir} \le 1 \;\implies\; \sum_i x_i w'_{ij} \le 1 \forall j.\]
\end{definition}

There are two main approaches one could envision to finding a $k$-dimensional gas measure:
\begin{enumerate}
    \item \textbf{Partitioning}: Partition the resources into $k$ groups and apply the single-dimensional measure to each group.
    \item \textbf{Factorization}: Find a low-rank factorization of $W'$ as $AB + E$ where $A$ and $B$ have small dimensions and $E$ is a small error term.
\end{enumerate}
We consider each in turn.

\subsection{A Partitional Approach}
Recall that in our single-dimensional model for \(n\) resources, for each primitive operation \(i\) we define the optimal gas cost as
\[
g_i = \max_{j\in [n]} \frac{w_{ij}}{B_j},
\]
where \(w_{ij}\) is the usage of operation \(i\) on resource \(j\) and \(B_j\) is the capacity of resource \(j\). Suppose we wish to reduce the \(n\) resources into \(k\) groups. Let 
$\mathcal{P} = \{P_1, P_2, \ldots, P_k\},$
be a partition of the resource set \(\{1,2,\ldots,n\}\) into \(k\) disjoint subsets. For each subset \(P_r\) we apply the single-dimensional measure only to the resources in \(P_r\). That is, for each operation \(i\) we define
\[
g_i^{(r)} = \max_{j\in P_r} \frac{w_{ij}}{B_j}.
\]
If a transaction \(x \ge 0\) uses the operations, then its aggregated cost with respect to subset \(P_r\) is 
\[
C_r(x) = \sum_{i} x_i \, g_i^{(r)}.
\]
A natural way to define the loss (or approximation factor) of a partition \(\mathcal{P}\) is to take the worst-case over the \(k\) subsets:
\[
L(\mathcal{P}) = \max_{r=1,\ldots,k} L(P_r),
\]
where \(L(P_r)\) is the loss incurred by using the single-dimensional measure on resources in \(P_r\). Our goal is to choose the partition \(\mathcal{P}\) that minimizes \(L(\mathcal{P})\).
For our purposes, we define the following decision version, i.e., does there exist a partition \(\mathcal{P}\) of the resources into \(k\) subsets such that
\[
L(\mathcal{P}) \le \alpha_0?
\]

We now show that even in the case \(k=2\) the decision problem is NP-complete by reducing from the well-known Equal Cardinality Partition \cite{GJ79}.

\begin{theorem}\label{thm:optimal-partition}
The decision problem for optimal \(k\)-partitioning for gas measures (with \(k=2\)) is NP-complete.
\end{theorem}

This reduction shows that when the goal is to partition \(n\) resources into \(2\) subsets such that the worst-case loss (as measured by the aggregated single-dimensional gas measures on each subset) is minimized, the problem is NP-complete. (A similar argument extends to any fixed \(k \ge 2\).) 

\subsection{A Factorization Approach}

This approach seeks factor matrices $A, B$ such that $AB$ provides an upper bound on $W'$.

\begin{theorem}[k-Dim. Representation via Upper-Bounding Factorization]
\label{th:kdim_sufficient}
Let $W' \in \Rnonneg^{|I|\times n}$ be the normalized operation-resource matrix. Suppose there exist matrices $A \in \Rnonneg^{|I|\times k}$ and $B \in \Rnonneg^{k\times n}$ such that:
\begin{enumerate}
    \item $W'_{ij} \le (AB)_{ij}$ for all $i, j$. (Element-wise inequality)
    \item The L1 norm of columns of $B$ is bounded: $\norm{B_{\cdot, j}}_1 = \sum_{r=1}^k B_{rj} \le 1$ for all $j=1, \dots, n$.
\end{enumerate}
Then, the $k$-dimensional gas costs given by $A_{ir}$ (i.e., $g_i = A_{i,\cdot}$) represents $W'$.
\end{theorem}

The challenge here lies in finding suitable matrices $A, B$ that satisfy $W' \le AB$ and $\norm{B_{\cdot, j}}_1 \le 1$: this is essentially Non-negative Matrix Factorization (NMF): see e.g.,  \cite{NIPS2000_f9d11525}. Trading off between $k$ (the number of dimensions) and the size of $B$ is a non-trivial problem. We leave this as an open problem for future work.

Nevertheless, given a factorization, we can use the associated matrix $A$ as a k-dimensional gas measure. It is then natural to ask the associated approximation factor.

\begin{theorem}\label{th:factorization_approx}
    Let $A$ be a $k$-dimensional gas measure for $W'$ calculated using the factorization approach of Theorem \ref{th:kdim_sufficient}. For each $\ell  \in [k]$, define the utility matrix $U^\ell$ as $U^\ell_{ij} = w'_{ij} / A_{i\ell}$. Then, the approximability is the reciprocal of the minimum game value among the games $\{\ell \in [k] \mid U^\ell \}$.
\end{theorem}

How then does this approach compare to the partitional approach we outlined previously? We show that the factorization approach is at least as good as the partitional approach: 
\begin{corollary}\label{cor:factorization_approx}
    For any number of dimesnions $k$, there exists a $k$-dimensional gas measure $A$ via the factorization approach of Theorem \ref{th:kdim_sufficient} that is at least as good as the optimal $k$-dimensional gas measure via the partion approach of Theorem \ref{thm:optimal-partition}.
\end{corollary}

Finding good algorithms for the factorization approach is an open problem, as we described above. However, in practice, there are only a small, finite number of resources, and so the factorization approach may nevertheless be useful via  brute-force. The exact gains to be had from using this approach is an interesting empirical problem that depends on the details of the blockchain and the architecture of its clients, and we leave this as future work.

\section{Practical Bounds}

The analysis above is a theoretical worst-case one.  Using it 
directly on the set of possible operations of a system
will often give you bad estimates. For example, when considering a single-dimensional gas measure. There will always be
{\em some} operations that use only a single resource, 
which takes us back to the worst case over all systems for
which the bound is as high as the number of different resources.
In reality one may get more ``mileage'' ---
i.e., better bounds on the maximal possible loss ---
by taking into account some further knowledge about the actual system.  We mention two significant such considerations.  Finally we conclude with a discussion of the modeling assumption that underlie our paper. 

\subsection{Using Distributional Information}\label{sec:practical}

Our analysis above was worst-case in two respects: not only does the gas measure $g$ always
guarantee  that we are within the constraints of the system, but the approximation factor
calculated was for a worst-case mixture of operations.  
Practically, we may have a reasonable estimate, based on historical data, of the actual mix of operations.  
A reasonable practical goal would be to choose a gas measure that is safe in the worst-case --- i.e., still take our 
minimal $g$ --- but evaluate its quality based on the historical mix of operations.   This is quite easy to do using the game specified, by using the historical 
distribution on operations as the row player's strategy $x^{hist}$  (after taking into account the implied scaling) and then looking at the best reply of the column (maximizing) player to it. Since the row (minimizing) player 
is not playing his optimal strategy anymore, the column player will be able to achieve a higher outcome 
$\nu^{hist} \ge \nu$ than the game's value, and the reciprocal of that $\alpha^{hist} = 1/\nu^{hist}$ will give a lower, better, approximation ratio.  As an extreme example of this, consider the operation-matrix in Table \ref{tab:hist} with two resources and three operations.
Another example that spells out how the operation 
frequencies are scaled to get the distribution $x^{hist}$ 
appears in Appendix \ref{app-dist}.

More generally, historical analysis may give a range of possible distributions. By defining the set of strategies of the row player to be this range, and computing the minimax, we have the worst case performance of the minimal measure $g$ over historical data. 
\begin{table}
\centering

\centering
\begin{tabular}{|c||c|c||c|} 
\hline 
{\bf Operation} & {\bf $w_{i1}$} & {\bf $w_{i2}$} & $g_i$\\ 
\hline \hline
{\bf Op1}  & 0 & 1 & 1\\ 
\hline 
{\bf Op2}  & 1 & 1 & 1\\ 
\hline 
{\bf Op3}  & 1 & 0 & 1\\ 
\hline \hline
{\bf $B_j$}    & 1 & 1 &\\ 
\hline
\end{tabular}
\caption{\small An extreme example with three operations and two resources that has an approximability ratio of 2 with
the row player's equilibrium strategy being $(1/2,0,1/2)$.  However if
the actual historical 
distribution is, say,
$(5\%,80\%,15\%)$, then the maximizing player
will play the first resource giving him expected utility of
$0.85$ implying that
the actual loss by using $g$ will only be $\alpha^{hist}=1/0.85 = 20/17$.}
\label{tab:hist}
\end{table}

An important observation is that
this approximation guarantee holds not just over a single block 
with the given distribution of operations
but even over a sequence of blocks with the given {\em average} distribution, and this is so since the
column player can consistently play his best-reply to the average distribution.

\subsection{Non-congesting Resources}

The analysis above need only be applied to resources that actually constrain the system 
at the block level.   Blockchains systems often have other resources that do not constrain the system at the block level and these can be taken care of in a way that is completely independent of our analysis above, and in particular do not need to be folded into the 
single gas measure.  This is significant since these resources are often rather ``orthogonal'' to other congesting resources so
folding them into the same single gas measure would incur a much higher loss.

Take for example a typical ``Layer 2'' blockchain.   The amount of storage cells modified by a transactions
carries a cost to the system (that needs to pay a data availability provider to store it),
but is not by itself a binding constraint on the system.  In such cases, the use of such a 
``non-congesting'' resource needs to be tracked and paid for but since these resource do not limit inclusion in the block they need not
be considered in the operation-resource matrix or folded into our single gas measure.

\ignore{
Another example are resources that do not constrain a single block but rather become
constraints only over long time scales.  An example of such a resource is the total amount
of new storage that is created and needs to be maintained ``forever''.  While this may
indeed become a very significant constraint in the long term, 
it has has no ``bite''
at the single block level.  A naive strategy for
handling such a long-term constraint is 
to translate it into a block-by-block constraint by dividing the long-term desired quota over the number of blocks in the desired 
long-term period of time.  
This not only is inefficient as it introduces significant artificial constraints\footnote{
This is, e.g.,  
    the current mechanism in Ethereum that simply assigns a high gas cost to
    the creation of a new storage cell and has the 
    drawback that storage prices artificially 
    fluctuate with the short-term load on Ethereum.  An extreme
    demonstration of the inefficiency that this causes was observed before 
    EIP-3529 \cite{3529} came into effect
    when there was a significant rebate (75\%) 
    for erasing a storage cell, and some users would ``buy'' 
    storage cells when gas prices
    were low and then ``sell'' them when prices were high.  
    This necessitated EIP-3529 \cite{3529} which significantly reduced
    the rebate amount (to 20\%), with the unfortunate side effect of 
    reduction in the incentives
    for practicing ``good hygiene'' of erasing storage cells that are 
    no longer needed.
}
but also adds a dimension - which increases the loss when folded into
a single gas measure.  An alternative approach would be to set a price
for the resource that will result in the long-term required bound on usage.
While in principle, this could be done by hand using trial and error, it is also possible to use a slow-moving EIP-1559-like mechanism to automatically - but slowly - change the price so that in the short-term
it will behave like a fixed cost that need not be folded into the single gas measure as above and, in the long-term provide the appropriate constraint.
We demonstrate this using an example in appendix \ref{app-state}.
} % ignore long time scales

\subsection{Modeling Assumptions}
There are two main modeling assumptions that we make in this paper:
\begin{enumerate}
    \item The resources used by a given operation are \emph{fixed} and known in advance.
    \item The resources used by a collection of  transactions are the sum of the individual resource usages.
\end{enumerate}

Both of these are sensible to a first approximation. For example, the amount of storage used by a collection of transactions is indeed the sum of the individual storage usages. 

That said, in modern high-performance blockchains, both assumptions may be less appropriate. For example, many blockchains use parallel processing to speed up the execution of transactions. In this case, the total computational usage of a transaction depends on the other transactions in the block (due to write conflicts etc), since this determines how parallelizable the block is. Similarly, the cost of reading from cold storage may depend on whether a recent past transaction has already read the same piece of data (in which case it is already in memory). In such cases, our resource matrix may be interpreted as a ``maximum'' usage. The actual usage may be lower, and using distributional assumptions as discussed in Section \ref{sec:practical} may give a better approximation.
\newpage
\bibliographystyle{alpha}
\bibliography{bib}

\appendix

\section{Appendix}

\subsection{Difficulties with Multidimensional Resources} \label{app-comp}

While it is certainly obvious that handling more complex scenarios 
such as multi-dimensional resources is more difficult than handling the
simpler case of a single-dimensional gas, it is still worthwhile some
of the specific difficulties that this entails.

We start with the algorithmic problem.  
In the multidimensional case, assembling the block becomes a multi-dimensional knapsack problem (see, 
e.g., survey \cite{F04}).  Not only is it NP-complete but even solving
the {\em linear programming relaxation} is no longer an easy 
greedy algorithm, 
and need not be as good a heuristic for the original program as it is in the single dimensional case.  Furthermore, looking at the assembly of the sequence of blocks as an {\em online problem}, 
the multidimensional case can 
provably not
be handled as well as the single dimensional case \cite{BN25}.  
Finally, transactions often lack {\em precise estimates}
of the quantities of resources
that they use.  In Ethereum, for example, only upper bounds are given, but the actual payment is 
according to the actual amount of resources used.  Luckily the basic
greedy algorithm for composing a block handles this discrepancy well 
since it just
sorts the transaction by bid-per-unit and so over-estimating your resource
use hardly changes the outcome.\footnote{Except in the rare case that your 
transaction just
fits into the remaining block space but the over-estimate makes it 
look as though it wouldn't.}  It is not clear that the same robustness to resource use declaration will hold in the multi-dimensional case.

We continue with the strategic difficulties. How would one even bid in the multi-dimensional case? The natural way would be to suggest a single bid for the whole transaction (which can specify the value of the transaction or a global tip amount, if using EIP-1559-like fixed prices) together with quantities for each resource.  This solution would rule out the natural semantics of bidding per-unit of resource (since it combines multiple resources into one bid) requiring precise estimates ----rather than just upper bounds---of the resources used by each transaction. The alternative is to provide a maximum price for each resource together with the quantity of use of this resource.  This option is not clear in its semantics and, in particular, destroys any``incentive compatibility'' as it is unclear how a user should split their value for a transaction into bids for the different resources. Ethereum currently uses a hybrid method that seems susceptible to both types of difficulties where a maximum price is specified  for each resource (gas and blobs) but a ``tip'' (per-unit) is specified only for the gas resource.

\subsection{Proofs} \label{app-proof}

\begin{proof}[Proof of Observation \ref{obs:measure}]
Look at some fixed operation $i^*$ and resource $j^*$ 
and consider $x_{i^*}=1$ and $x_{i}=0$ for all $i\ne i^*$. 
Clearly, in order for $\sum_i x_i g_i  = g_{i^*} \le 1$ to imply 
$\sum_i x_i w_{ij^*} = w_{i^*j^*} \le B_{j^*}$ we must have $g_{i^*} \ge w_{i^*j^*}/B_{j^*}$.
Taking the maximum over all possible resources, we must have 
$g_{i^*} \ge \max_j w_{i^*j}/B_{j}$.  To see that this gas measure suffices, 
assume that  $\sum_i x_i \max_j w_{ij}/B_{j} \le 1$ we thus have that for every $j$:
$\sum_i x_i \cdot w_{ij}/B_{j} \le 1$ and thus $\sum_i x_i \cdot w_{ij} \le B_{j}$, as needed.  
\end{proof}

\bigskip

\begin{proof}[Proof of Theorem \ref{th:main}]
    Let $\nu$ be the value of this game and denote $\alpha = 1/\nu$. Let $(x^*_i)$ be the equilibrium strategy of the row (operation) minimizing player.

    The first thing that we have to show is that for any non-negative vector $x=(x_i)$ 
    we have that 
    \[\forall j:\;  \sum_i x_i w_{ij} \le B_j \; \implies \; \sum_i x_i g_i \le \alpha.\]  
    Assume by way of contradiction that $\sum_i x_i g_i = \alpha' > \alpha$. Define $x'_i = x_i g_i /\alpha'$. So we have that $\sum_i x'_i = 1$ so $(x'_i)$ is a valid strategy for the row player. Therefore, we must have that
    \[\max_j \sum_i x'_i \cdot u_{ij} \ge \nu,\] 
    i.e., for some $j$ we have that,
    \[\sum_i (x_i\cdot \frac{g_i}{\alpha'}) \cdot \frac{w_{ij}}{B_j \cdot g_i} \ge \nu.\] 
    Since $\alpha' \cdot \nu >1$ we get $\sum_i x_i  \cdot w_{ij} > B_j$, a contradiction, as needed.

    The second thing that we have to show is that $\alpha$ is optimal in the sense that  there exists some non-negative vector $x=(x_i)$ with $\sum_i x_i w_{ij} \le B_j$ such that for some $j$ we have $\sum_i x_i g_i \ge \alpha$.  Take $x_i = \alpha \cdot x^*_i / g_i$ where $x^*_i$ is the minimizing strategy of the row player. Note that for this choice of $x_i$ we have that, $\sum_i x_i g_i = \alpha \cdot \sum_i x^*_i = \alpha$. Further, note that:
     \[\sum_i x_i w_{ij} = \alpha B_j \cdot \sum_i x^*_i w_{ij} / (B_j \cdot g_i) =
    \alpha B_j \cdot \sum_i x^*_i u_{ij}.\]  Since $x^*$ is the minimizing strategy
    we have that for every $j$, $\sum_i x^*_i u_{ij} \le \nu$
    and thus  $\sum_i x_i w_{ij} \le \alpha B_j \nu = B_j$ as needed. 
\end{proof}

\bigskip

\begin{proof}[Proof of Theorem \ref{thm:optimal-partition}]
We prove the theorem by reduction from Equal Cardinality Partition (ECP), which is a well-known NP-complete problem.

\begin{definition}[Equal Cardinality Partition (ECP)]
Given an integer $k$ and a multiset $S$ of $2k$ positive integers that sum to $2T$, the ECP problem is to partition $S$ into $2$ subsets of cardinality $k$ that each sum to $T$.
\end{definition}

We will show that the 2-gas-partition problem is NP-complete by reduction from ECP.  Fix an instance of ECP where $S = \{s_1, \dots, s_{2k}\}$.  We construct an instance of the 2-gas-partition problem defined as:
\begin{itemize}
    \item \textbf{Operations}: There are two operations for each element $s_i \in S$, i.e. $|I| = 4k$: for each $\ell \in [2k]$, we have two operations $i_\ell^1, i_\ell^2$.
    \item \textbf{Resources}: There are two resourses for each element $s_i \in S$, i.e. $|J| = 4k$: for each $\ell \in [2k]$, we have two resources $j_\ell^1, j_\ell^2$.
    \item \textbf{Costs}: For each $\ell \in [2k]$, operations $i_\ell^1, i_\ell^2$ use $0$ of all resources except $j_\ell^1, j_\ell^2$. The 2-by-2 submatrix of these costs is given by:
    \[
        \begin{pmatrix}
            1 & 1- \kappa_\ell  \\
            1-\kappa_\ell & 1 
        \end{pmatrix}
    \]
    where $\kappa_\ell = 2 s_\ell \epsilon / (1+ s_\ell \epsilon)$ and $\epsilon$ is any constant such that $0 < \epsilon < \frac{1}{2T}$.
\end{itemize}

\begin{claim}
    The instance of ECP has a solution if and only if the instance of the 2-gas-partition problem has a solution of approximability $k +  T\epsilon$.
\end{claim}

\noindent
Consider the following proof of the claim for each direction:

\bigskip
\noindent
    ($\implies$) Suppose that $S$ can be partitioned into two subsets $S_1, S_2$ of cardinality $k$ such that $\sum_{s \in S_1} s = \sum_{s \in S_2} s = T$. Consider the corresponding parition of the resources into the two subsets $J_1, J_2$ of cardinality $2k$ each, where $s_\ell \in S_i \implies j_\ell^1, j_\ell^2 \in J_i$ for $i \in \{1,2\}$. 

    The minimax strategy for the row player in game $i$ is to play the operations $i_\ell^1, i_\ell^2$ each with probability:
    \begin{align*}
       & \frac12 \frac{(1+ s_\ell \epsilon)}{\sum_{s \in S_i} (1+ s \epsilon)} \\
       =& \frac12 \frac{(1+ s_\ell \epsilon)}{k + T\epsilon}.
    \end{align*}
    and the minimax strategy for the column player is to randomize over all resources in $J_i$ with equal probability.

    The value of the game is the $1/ (k + T\epsilon)$ and the approximability of the 2-gas-partition problem is thus $k +  T\epsilon$ as needed.

    \bigskip
    \noindent
    ($\impliedby$) Suppose that there exists a partition of the resources into two subsets $J_1, J_2$, such that the 2-gas-partition problem has a solution of approximability $k +  T\epsilon$. This implies that each game $i$ has a solution of approximability $(k +  T\epsilon)$. We need to show that this implies that $S$ can be partitioned into two subsets $S_1, S_2$ of cardinality $k$ such that $\sum_{s \in S_1} s = \sum_{s \in S_2} s = T$.

    First note that each of $J_1$ and $J_2$ must have cardinality exactly $2k$. To see this, note that otherwise the larger set would necessarily have an approximability of at least $k + 1$ which is larger than $k +  T\epsilon$.

    Next, note that the resources in $J_1$ must be paired (and therefore, also, the resources in $J_2$). To see this, suppose not, suppose that $J_1$ contains $m>0$ unpaired resources and $k-\frac{m}{2}$ pairs of paired resources (our previous claim implies that $m$ must be even). Then, the approximability of each game is at least $k + m$ which is greater than $k +  T\epsilon$ by our assumption.

    Finally, suppose that each of $J_1$ and $J_2$ contain exactly $k$ pairs of paired resources. Denote the subset of $S$ that corresponds to the resources in $J_1$ as $S_1$, and the subset of $S$ that corresponds to the resources in $J_2$ as $S_2$. Observe that the for the game corresponding to $J_1$, the corresponding minimax strategy if for the row player to play the operations $i_\ell^1, i_\ell^2$ each with probability:
    \begin{align*}
        & \frac12 \frac{(1+ s_\ell \epsilon)}{\sum_{s \in S_1} (1+ s \epsilon)} \\
        =& \frac12 \frac{(1+ s_\ell \epsilon)}{k + \epsilon \sum_{s \in S_1} s }.
    \end{align*}
    The approximability of this game is thus $k +  \epsilon \sum_{s \in S_1} s $, and similarly for $S_2$.  Therefore if the approximability of the 2-gas-partition problem is $k +  T\epsilon$, we must have that $\sum_{s \in S_1} s = \sum_{s \in S_2} s = T$.
\end{proof}
\bigskip

\begin{proof}[Proof of Theorem \ref{th:kdim_sufficient}]
Assume $x \ge 0$ satisfies $\sum_i x_i A_{ir} \le 1$ for all $r$. Let $y_r = \sum_i x_i A_{ir}$, so $0 \le y_r \le 1$.
We want to show $\sum_i x_i w'_{ij} \le 1$ for all $j$.
Using condition (1) ($W' \le AB$):
\[ \sum_i x_i w'_{ij} \le \sum_i x_i (AB)_{ij} \]
Expanding the matrix product:
\[ \sum_i x_i (AB)_{ij} = \sum_i x_i \sum_r A_{ir} B_{rj} = \sum_r \left( \sum_i x_i A_{ir} \right) B_{rj} = \sum_r y_r B_{rj} \]
Since $y_r \le 1$ and we assumed $B_{rj} \ge 0$:
\[ \sum_r y_r B_{rj} \le \sum_r (1) B_{rj} = \sum_r B_{rj} = \norm{B_{\cdot, j}}_1 \]
Using condition (2) ($\norm{B_{\cdot, j}}_1 \le 1$):
\[ \sum_r y_r B_{rj} \le 1 \]
Combining the inequalities, we have $\sum_i x_i w'_{ij} \le 1$ for all $j$.
\end{proof}

\bigskip

\begin{proof}[Proof of Theorem \ref{th:factorization_approx}]
The proof is similar to the proof of Theorem \ref{th:main}.
Fix any $\ell \in [k]$ and let $U^\ell$ be the utility matrix for the game corresponding to the $\ell$-th dimension. Let $\nu^\ell$ be the value of this game and let $\alpha^\ell = 1/\nu^\ell$.

The first thing that we have to show is that for any non-negative vector $x=(x_i)$ 
    we have that 
    \[\forall j:\;  \sum_i x_i w'_{ij} \le 1 \; \implies \; \sum_i x_i A_{i \ell} \le \alpha^\ell.\]  
    Assume by way of contradiction that $\sum_i x_i A_{i \ell} = \alpha' > \alpha^\ell$. Define $x'_i = x_i A_{i \ell} /\alpha'$. So we have that $\sum_i x'_i = 1$ so $(x'_i)$ is a valid strategy for the row player. Therefore, we must have that
    \[\max_j \sum_i x'_i \cdot u_{ij}^\ell \ge \nu^\ell,\] 
    i.e., for some $j$ we have that,
    \[\sum_i (x_i\cdot \frac{A_{i \ell}}{\alpha'}) \cdot \frac{w'_{ij}}{1 \cdot A_{i \ell}} \ge \nu^\ell.\] 
    Since $\alpha' \cdot \nu^\ell >1$ we get $\sum_i x_i  \cdot w'_{ij} > 1$, a contradiction, as needed.

    The second thing that we have to show is that $\alpha^\ell$ is optimal in the sense that  there exists some non-negative vector $x=(x_i)$ with $\sum_i x_i w'_{ij} \le 1$ such that for some $j$ we have $\sum_i x_i A_{i \ell} \ge \alpha^\ell$.  Take $x_i = \alpha^\ell \cdot x^*_i / A_{i \ell}$ where $x^*_i$ is the minimizing strategy of the row player. Note that for this choice of $x_i$ we have that, $\sum_i x_i A_{i\ell} = \alpha^\ell \cdot \sum_i x^*_i = \alpha^\ell$. Further, note that:
     \[\sum_i x_i w'_{ij} = \alpha^\ell \cdot \sum_i x^*_i w'_{ij} / (A_{i\ell}) =
    \alpha^\ell \cdot \sum_i x^*_i u_{ij}^\ell.\]  Since $x^*$ is the minimizing strategy
    we have that for every $j$, $\sum_i x^*_i u_{ij}^\ell \le \nu^\ell$
    and thus  $\sum_i x_i w'_{ij} \le \alpha^\ell \nu^\ell = 1$ as needed. 

    Having shown the approximation factor for any single dimension, the approximation factor for the $k$-dimensional gas measure is clearly the maxmium, i.e., the reciprocal of the minimum game value almong the games $\{\ell \in [k] \mid U^\ell \}$.
\end{proof}

\bigskip

\begin{proof}[Proof of Corollary \ref{cor:factorization_approx}]
Consider the partitional approach and fix any $k$-element partition of the resources into $k$ subsets $J_1, \dots, J_k$. Note that we could have also written the associated gas measure as one potential $k$-dimensional factoring where: 
\begin{align*}
    A_{i \ell} &= \max_{j \in J_\ell} w'_{ij}, \\
    B_{\ell j} &= \mathbbm{1}_{j \in J_\ell}.
\end{align*}
Therefore clearly the factorization approach is at least as good as the partitional approach.
\end{proof}

\bigskip

\subsection{An example for using distributional information of operations}\label{app-dist}

Consider the system described in Table 1, with its minimal
gas measure $g$ denoted in the table too.   Assume that we know that all four operations 
have equal frequency $f_1=f_2=f_3=f_4=1/4$.  
To translate these frequencies into the probabilities played by
the row player in the game depicted in Table 2, we need to follow the two implied 
scalings that are depicted in tables \ref{tab:4} and \ref{tab:5}.  The first
scaling normalizes all columns by dividing each column by $B_j$, obtaining table \ref{tab:4} and 
the second scaling normalizes all rows by dividing each row by $g_i$ obtaining table
\ref{tab:5} which gives exactly the game shown in table \ref{tab:my_label2}.  
The row player's probability for each of the normalized operations is normalized
by multiplying $f_i$ by $g_i$ and then normalizing the sum of the normalized frequencies to 1.
Thus for example the translated frequency of the normalized first operation will be calculated as 
$x^{hist}_1 = f_1 \cdot g_1 /(\sum f_i g_i) = (1/4)\cdot(1/3)/((1/3+2/3+3/5+2/3)/4) = 5/34$ 
and similarly, $x^{hist}_2 = 10/34$, 
$x^{hist}_3 = 9/34$ and $x^{hist}_4 = 10/34$. Consequently, playing column 1 gives $(5/34)\cdot(2/5)+(10/34)\cdot(3/5)+(9/34)\cdot 1+(10/34)\cdot 1 =27/34$ and playing column 2 gives $25/34$. So the loss in this given operation distribution is $34/27$ (reciprocal of the obtained utility), 
which is lower than the loss in the completely worst-case bound of 11/8.  

\begin{table}
\centering
\begin{minipage}{0.45\textwidth}
\centering
\begin{tabular}{|c|c||c|c||c|} 
\hline 
{\bf Freq.} & {\bf OP} & {\bf $w_{i1}$} & {\bf $w_{i2}$} & $g_i$\\ 
\hline \hline
1/4 & {\bf Op1}  & 2/15 & 1/3 & 1/3\\ 
\hline 
1/4 & {\bf Op2}  & 2/5 & 2/3 & 2/3\\ 
\hline 
1/4 & {\bf Op3}  & 3/5 & 1/3 & 3/5\\ 
\hline 
1/4 & {\bf Op4}  & 2/3 & 1/3 & 2/3\\ 
\hline \hline
1 & {\bf $B_j$}    & 1 & 1 &\\ 
\hline
\end{tabular}
\caption{\small The same example as in table \ref{tab:my_label1}, after normalizing $B_j$'s to 1, with the
raw operation frequencies.}
\label{tab:4}
\end{minipage}\hfill
\begin{minipage}{0.45\textwidth}
\centering
\begin{tabular}{|c|c||c|c||c|} 
\hline 
{\bf Freq.} & {\bf OP} & {\bf $w_{i1}$} & {\bf $w_{i2}$} & $g_i$\\ 
\hline \hline
1/12 & {\bf Op1}  & 2/5 & 1 & 1\\ 
\hline 
1/6 & {\bf Op2}  & 3/5 & 1 & 1\\ 
\hline 
3/20 & {\bf Op3}  & 1 & 5/9 & 1\\ 
\hline 
1/6 & {\bf Op4}  & 1 & 1/2 & 1\\ 
\hline \hline
17/30 & {\bf $B_j$}    & 1 & 1 &\\ 
\hline
\end{tabular}
\caption{\small The same example as in table \ref{tab:4}, after normalizing $g_i$'s to 1.  The frequences are also adjusted.  E.g. each normalized-OP1 is equivalent to 3 original-OP1's.}
\label{tab:5}
\end{minipage}
\end{table}

\end{document}